\documentclass[reqno]{amsart}

\usepackage{a4,amsmath,amssymb,amsthm}
\usepackage{xfrac}

\usepackage{amsaddr}

\usepackage[T1]{fontenc}
\usepackage[section]{placeins}

\makeatletter
  \def\moverlay{\mathpalette\mov@rlay}
  \def\mov@rlay#1#2{\leavevmode\vtop{%
     \baselineskip\z@skip \lineskiplimit-\maxdimen
     \ialign{\hfil$#1##$\hfil\cr#2\crcr}}}
\makeatother

\usepackage{stmaryrd}

\usepackage{url}
\usepackage{hyperref}

\usepackage{color}

\numberwithin{equation}{section}

\usepackage[mathscr]{eucal}

\usepackage{graphicx}
\graphicspath{pic/}
\usepackage{psfrag}
\usepackage{blindtext}
\usepackage[inline]{enumitem}
\usepackage[labelfont=bf]{subfig}

\newcommand{\HH}{\mathcal H}

\newcommand{\scri}{\mathcal{I}}

\theoremstyle{plain}
\newtheorem{thm}{Theorem}
\newtheorem{cor}[thm]{Corollary}
\newtheorem{lemma}[thm]{Lemma}

\newtheorem{definition}[thm]{Definition}

\newtheorem{remark}[thm]{Remark}
\theoremstyle{definition}
\newtheorem{case}{Case}

\newcounter{mnotecount}[section]

\newcounter{mymnotecount}[section]

\title{Smoothness of the future and past trapped sets in Kerr-Newman-Taub-NUT spacetimes}
\author[C.~F. Paganini and M.~A. Oancea]{ Claudio F. Paganini$^\dagger$,   and Marius A. Oancea$^\dagger$} 
\email{claudio.paganini@aei.mpg.de}
\email{marius.oancea@aei.mpg.de}
\address{$^\dagger$Albert Einstein Institute, Am M\"uhlenberg 1, D-14476 Potsdam,  Germany }

\begin{document}

\date{\today \ {\em File:\jobname{.tex}}}

\begin{abstract}
We consider the sets of future/past trapped null geodesics in the exterior region of a sub-extremal Kerr-Newman-Taub-NUT spacetime. We show that from the point of view of any timelike observer outside of such a black hole, trapping can be understood as two smooth sets of spacelike directions on the celestial sphere of the observer.
\end{abstract}

\maketitle

\tableofcontents

\section{Introduction}The trapped sets on the celestial sphere are closely related to the notion of black hole shadows. The shadow of the black hole is defined as the innermost trajectory on which light from a background source passing a black hole can reach the observer. The past trapped set of null geodesics through a point thus corresponds to the boundary of the black hole shadow. The first discussion of the shadow in Schwarzschild spacetimes can be found in \cite{synge_escape_1966}, and, for extremal Kerr at infinity, it was later calculated in \cite{bardeen_black_1973}. Analyzing the shadows of black holes is of direct physical interest as there is hope for the Event Horizon Telescope to be able to resolve the black hole in the center of the Milky Way well enough so that one can compare it to the predictions from theoretical calculations, see for example \cite{doeleman_event-horizon-scale_2008}. This perspective has led to a number of advancements in the theoretical treatment of black hole shadows in recent years \cite{cunha_shadows_2016,grenzebach_aberrational_2015,grenzebach_photon_2014,grenzebach_photon_2015,hioki_measurement_2009,li_measuring_2014}. 

Our main contribution in the present paper is the proof of Theorem \ref{thm:1} which is the first rigorous proof of the observations in \cite{grenzebach_photon_2014}. The significance of Theorem \ref{thm:1} is that we prove that for any subextremal Kerr-Newman-Taub-NUT spacetime, including Schwarzschild, the past and the future trapped sets at any regular point in the exterior region are smooth closed curves on the celestial sphere of any observer. We would like to stress that Theorem \ref{thm:1} therefore describes a property of trapping which does not change when going from Schwarzschild to Kerr-Newman-Taub-NUT.

Beyond its relevance for the discussion of black hole shadows the result is also of interest with respect to decay estimates for fields in the exterior region of Kerr black holes. Trapping is one of the biggest obstacle to prove such decay results. The area of trapping changes substantially when going from Schwarzschild, where it is restricted to one fix radius, to Kerr where the area of trapping covers a finite range of radii. This makes it a lot harder to prove decay in Kerr than it is for Schwarzschild. This difficulty was only recently overcome in \cite{2014arXiv1402.7034D} for all subextremal Kerr space times. To study the decay of fields in Kerr, it is thus important to understand which properties of trapping survive when going from Schwarzschild to Kerr.  

The main result of the present work will be used in our upcoming work \cite{shadows} where we discuss the maximum amount of information that can be extracted from the shape of the shadows. This paper is an extraction of the novel material contained in our lecture notes \cite{oldpaper} with an extension to the Kerr-Newman-Taub-NUT class of spacetimes.

\subsection*{Overview of this paper} In section \ref{sec:Kerr-Newman-Taub-NUT} we collect some background on the Kerr-Newman-Taub-NUT spacetime. 
In section \ref{sec:geodeq} we discuss the geodesic equations in its separated form. In section \ref{sec:sphere} we prove our main Theorem on the topological structure of the past and future trapped sets. 

\section{The Kerr-Newman-Taub-NUT Spacetime} \label{sec:Kerr-Newman-Taub-NUT} 
The Kerr-Newman-Taub-NUT family of spacetimes describes axially symmetric and stationary black hole solutions to the Einstein-Maxwell field equations.  We use Boyer-Lindquist (BL) coordinates ($t, r,\phi,\theta$), which have the property that the metric components are independent of $\phi$ and $t$. The metric has the form \cite{griffiths2009}:
\begin{equation}
\label{eq:metric}
\begin{split}
    \mathrm{d}s^2 =& \Sigma \left(\frac{1}{\Delta}\mathrm{d} r^2 + \mathrm{d}\theta^2 \right)+ \frac{1}{\Sigma}\left((\Sigma+ a \chi)^2 \sin^2(\theta) -\Delta \chi^2\right)\mathrm{d} \phi^2\\
    & \frac{2}{\Sigma}\left(\Delta\chi-a(\Sigma+a\chi)\sin^2(\theta)\right)  \mathrm{d}t\mathrm{d}\phi- \frac{1}{\Sigma}\left(\Delta-a^2\sin^2(\theta)\right)\mathrm{d}t^2,
 \end{split}
\end{equation}
 where
\begin{equation}
\Sigma = r^2 + (l+a \cos  \theta)^2,
\end{equation}
\begin{equation}
\chi=a\sin^2(\theta)-2 l (\cos(\theta)+C),
\end{equation}
\begin{equation}
\Delta (r)= r^2 - 2Mr + a^2-l^2+Q^2.
\end{equation}
Here, $M$ is the mass, $a$ is the angular momentum, $Q^2 = Q^2_e + Q^2_m$ is the charge, and $l$ is the NUT parameter. Also, there is another parameter, C, which was first introduced by Manko and Ruiz \cite{Manko-Ruiz}. This parameter modifies the singularities on the rotation axis, introduced by the Taub-NUT parameter $l$. For a discussion of the physical nature of these parameters and the nature of the singularities in these spacetimes see \cite{grenzebach_photon_2014}, and for a discussion of the rotation axis in the case $l\neq0$ see \cite{TaubNUTsingularities}. In our present work we will stay away from any irregular points in these spacetimes, in particular those that arise for $l\neq0$ on the rotation axis. \\The zeros of $\Delta(r)$ are given by:
\begin{equation}
r_\pm=M\pm\sqrt{M^2-a^2+l^2- Q^2}
\end{equation}
and correspond to the location of the event horizon at $r=r_+$ and of the Cauchy horizon at $r=r_-$. In the present work we are only interested in the exterior region of the black hole spacetime, hence $r\in(r_+,\infty)$.
 For our considerations it is useful to introduce an orthonormal tetrad. A convenient choice is:
\begin{subequations}\label{eq:tetrad}
\begin{align}
e_0 &= \left.\frac{(\Sigma + a \chi)\partial_t+ a\partial_\phi}{\sqrt {\Sigma\Delta(r)}}\right|_p ,&\qquad e_1&=\left.\sqrt{\frac{1}{\Sigma}}\partial_\theta\right|_p, \\ \nonumber
e_2&=\left.\frac{-(\partial_\phi+\chi \partial_t)}{\sqrt{\Sigma}\sin(\theta)}\right|_p,&\qquad e_3&=\left.-\sqrt{\frac{\Delta(r)}{\Sigma}} \partial_r\right|_p.
\end{align}
\end{subequations}
This frame is a natural choice as the principal null directions can be written in the simple form $e_0 \pm e_1$. 
\section{Geodesic Equations}\label{sec:geodeq}
We now focus our attention on null geodesics. The constants of motion for geodesics in Kerr-Newman-Taub-NUT spacetimes  are the mass, which we set equal to zero, the energy, the angular momentum with respect to the rotation axis of the black hole and Carter's constant \cite{carter_global_1968}:
\begin{subequations}\label{eq:com}
\begin{align}
0&=g_{\mu\nu}\dot \gamma^\mu\dot \gamma^\nu , \\
E&= - (\partial_t)^\nu\dot \gamma_\nu, \\
L_z&= (\partial_\phi)^\nu \dot \gamma_\nu, \\
K &= \sigma _ {\mu \nu} \dot \gamma ^ {\mu} \dot \gamma ^ {\nu},
\end{align}
\end{subequations} 
where $\sigma_{\mu\nu}$ is a Killing tensor given by, cf. \cite{PhysRevD.76.084036}:
\begin{equation}
\sigma_{\mu \nu} = \Sigma ( (e_1)_\mu (e_1)_\nu + (e_2)_\mu (e_2)_\nu ) - (l + a \cos(\theta))^2 g_{\mu \nu}
\label{eq:Killing}
\end{equation}
The constants of motion can be used to decouple the geodesic equation to a set of four first order ODEs, e.g. \cite[p. 242]{MR1647491}:
\begin{subequations}\label{eq:eom}
\begin{align}
 \dot t &= \frac{\chi(L_z-E \chi)}{\Sigma \sin^2(\theta)}+\frac{(\Sigma+ a\chi)((\Sigma+a \chi)E-a L_z)}{\Sigma\Delta(r)},\\
\dot \phi &=\frac{L_z-E \chi}{\Sigma \sin^2(\theta)}+\frac{a((\Sigma+a \chi)E-a L_z)}{\Sigma\Delta(r)} ,\\
\Sigma ^2 \dot r ^2  &= R (r,E,L_z,K)= ((\Sigma+a \chi)E-aL_z)^2-\Delta(r) K,
\label{eq:radial}\\
\Sigma ^2 \dot \theta ^2 &= \Theta (\theta, E, L_z, Q) =  K -  \frac{(\chi E-L_z)^2}{\sin ^2 \theta}  ,
\label{eq:theta}
\end{align}
\end{subequations}
where the dot denotes differentiation with respect to the affine parameter $\lambda$. \footnote{The radial and the angular equation can be entirely decoupled by introducing a new non-affine parameter $\kappa$ for the geodesics. It is defined by $\frac{\mathrm{d} \kappa}{\mathrm d \lambda}=\frac{1}{\Sigma}$.} For $E \neq 0$ the four equations are homogeneous in $E$. For the radial and the angular equations we have: 
\begin{align}
R(r,E,L_z,K)&=E ^2 R(r,1,L_E,K_E),\\
\Theta (\theta, E, L_z, K)& =E ^2 \Theta(r,1,L_E,K_E).
\end{align} 
where $L_E=L_z/E$ and $K_E=K/E^2$.
One of the most important features of geodesic motion in black hole spacetimes is the possibility of trapping. A geodesic is called trapped if its motion is bounded in a spatially compact region away from the horizon. In Kerr-Newman-Taub-NUT this corresponds to the geodesics motion being bounded in the $r$ direction. In \cite{grenzebach_photon_2014} it was shown that one can obtain a parametrization of the conserved quantities for trapped null geodesics in terms of their radial location:
\begin{subequations}\label{eq:altcomtrap}
\begin{align}
K_E=\left.\frac{16 r^2 \Delta(r)}{(\Delta'(r))^2}\right|_{r=r_{trapp}},\\
aL_E=\left.(\Sigma +a \chi)-\frac{4r \Delta(r)}{\Delta'(r)}\right|_{r=r_{trapp}},
\end{align}
\end{subequations}
where $ \Delta'(r)=2 r - 2 M  $ is the partial derivative of $\Delta(r)$ with respect to $r$. Further the following inequality for the area of trapping was derived in \cite{grenzebach_photon_2014}:
\begin{equation}\label{eq:areaoftrapp}
(4r\Delta(r)-\Sigma \Delta'(r))^2\leq 16 a^2r^2\Delta(r)\sin^2(\theta). 
\end{equation}
Further it was shown that all spherical null geodesics in the exterior region are unstable and therefore no non-spherical trapped null geodesics can exist in this region. Thus, the above set is complete, in the sense that it includes all trapped null geodesics that exist.

\section{Trapping as a Set of Directions}\label{sec:sphere}
 In this section we will  introduce a formal framework for our discussion. This allows us to give a more technical discussion of the trapped sets in Kerr-Newman-Taub-NUT spacetimes. 
\subsection{Framework}
First we have to introduce the basic framework and notations. Let $\mathcal{M}$ be a smooth manifold with Lorenzian metric $g$. At any point $p$ in $\mathcal{M}$ it is possible to find an orthonormal basis $(e_0,e_1,e_2,e_3)$ for the tangent space, with $e_0$ being the timelike direction. It is sufficient to treat only future directed null geodesics as the past directed ones are identical up to a sign flip in the parametrization and we are only interested in global properties of the null geodesics. The tangent vector to any future pointing null geodesic can be written as:
\begin{equation}\label{eq:tangentplus}
\dot\gamma( k|_p)|_p =\alpha ( e_0+ k_1 e_1+ k_2 e_2+ k_3 e_3)
\end{equation} 
where $\alpha= -g(\dot\gamma,e_0)$ and $k =(k_1,k_2,k_3)$ satisfies $|k|^2 =1$, hence $k\in S^2$. The geodesic is independent of the scaling of the tangent vector as this corresponds to an affine reparametrization for the null geodesic. We will therefore set $\alpha=1$ in the following discussion. 
\begin{remark}We will refer to the $S^2$ as the celestial sphere of a timelike observer at $p$, whose tangent vector is given by $e_0$, along the lines of e.g. \cite[p.8]{penrose_spinors_1987}.\end{remark}
Given we fixed a starting point $p$ and a tangent vector \eqref{eq:tangentplus} by choosing $k$ and $\alpha$, there exists a unique solution to the geodesic equation with this initial data. Thus, we make the following definition:
\begin{definition}Let $\gamma (k|_p)$ denote a null geodesic through $p$ whose tangent vector at $p$ is given by equation \eqref{eq:tangentplus}.
\end{definition}
\noindent Suppose now that $\mathcal{M}$ is the exterior region of a black hole spacetime with a complete future and past null infinity $\scri^\pm$ and a boundary given by the future and past event horizon $\HH^+\cup\HH^-$. We can then define the following sets on $S^2$ at every point $p$.

\begin{definition}\label{def:futinf}
The future infalling set: $\Omega_{\mathcal{H}^+}(p):=  \{k\in S^2 | \gamma(k|_p)\cap\mathcal{H}^+ \neq \emptyset \} $.\\
The future escaping set: $\Omega_{\mathcal{I}^+}(p):=  \{k\in S^2| \gamma(k|_p)\cap\mathcal{I}^+ \neq \emptyset \}$ .\\
The future trapped set: $\mathbb{T}_+(p) := \{k\in S^2 | \gamma(k|_p)\cap(\mathcal{H}^+\cup \mathcal{I}^+)  = \emptyset \}$.\\
The past infalling set: $\Omega_{\mathcal{H}^-}(p):=  \{k\in S^2 | \gamma(k|_p)\cap\mathcal{H}^- \neq \emptyset \}$.\\
The past escaping set: $\Omega_{\mathcal{I}^-}(p):= \{k\in S^2 | \gamma(k|_p)\cap\mathcal{I}^- \neq \emptyset \}$.\\
The past trapped set: $\mathbb{T}_- (p):= \{k\in S^2 | \gamma(k|_p)\cap(\mathcal{H}^-\cup \mathcal{I}^-)  = \emptyset \}$
\end{definition}
\noindent We further define the trapped set to be: 
\begin{definition}
The trapped set: $\mathbb{T}(p):= \mathbb{T}_+(p)\cap\mathbb{T}_-(p)$.
\end{definition}
\noindent The region of trapping in the manifold $\mathcal{M}$ is then given by:
\begin{definition}
Region of trapping: $\mathcal{A}:= \{p\in \mathcal{M}| \mathbb{T}(p)\neq \emptyset\}$.
\end{definition}
\begin{figure}[t!]
\centering
\centering
  \subfloat[\label{fig:causal1}]{%
    \includegraphics[width=.47\textwidth]{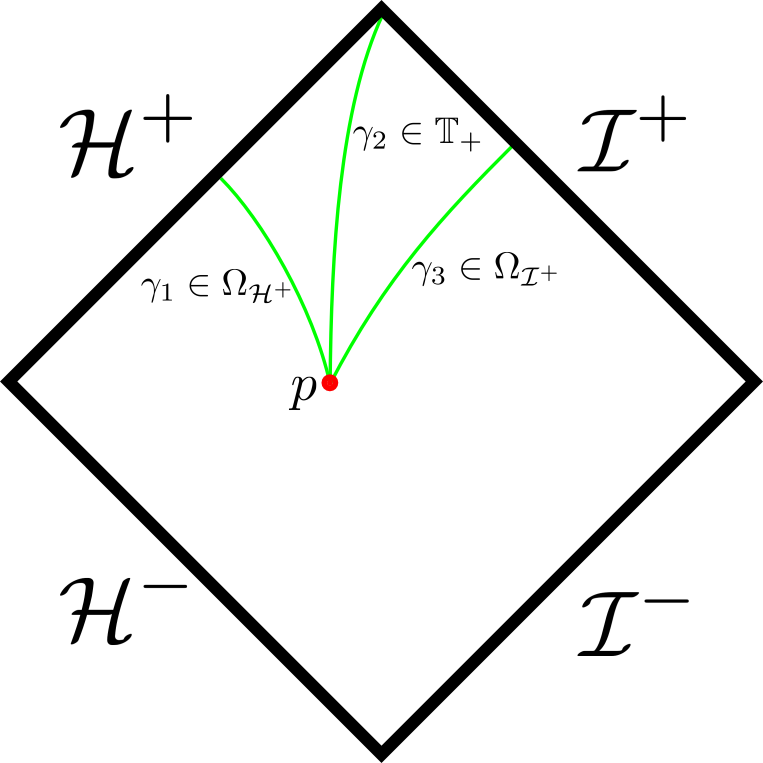}}\hfill
  \subfloat[\label{fig:causal2}]{%
    \includegraphics[width=.47\textwidth]{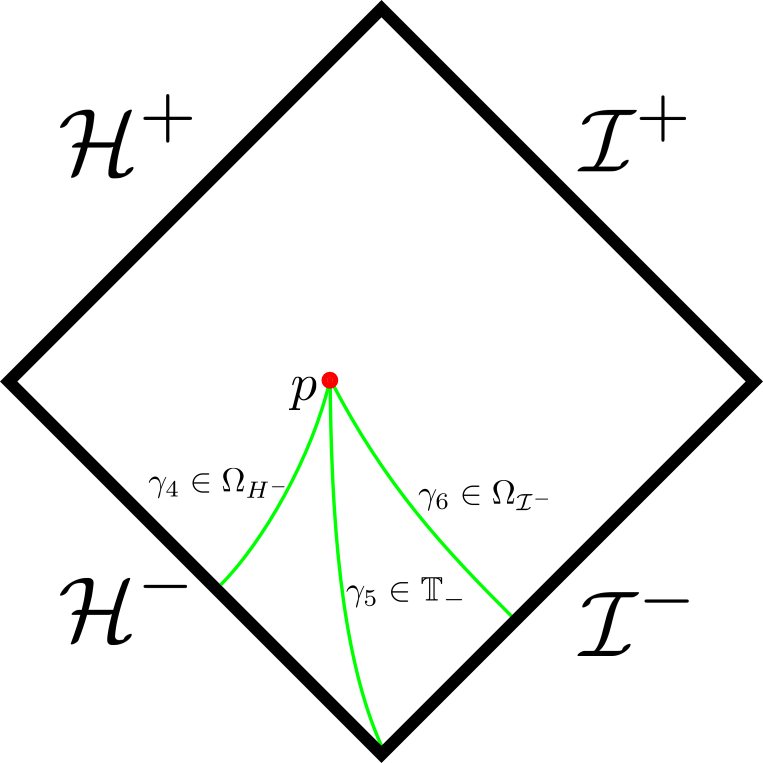}}\hfill
\caption{Conformal diagrams giving a schematic representation of elements of the sets in Definition \ref{def:futinf}. }
\label{fig:causal}
\end{figure}

\begin{definition}
We refer to the set $\Omega_{\mathcal{H}^-}(p)\cup \mathbb{T}_-(p)$ as the shadow of the black hole. 
\end{definition}
\noindent Suppose we place light sources close to infinity and thus certainly further out than the observer itself. The observer in the exterior region can then only observe light from such background sources along directions $k\in \Omega_{\mathcal{I}^-}(p)$ and hence the domain on the observers celestial sphere that belongs to the shadow is always going to be black. Therefore in any practical purposes an observer can only obtain information about the boundary of the shadow. In the the next section we will proof that this boundary is given exactly by $\mathbb{T}_-(p)$. 
\subsection{The trapped sets}\label{sec:Kerr-Newman-Taub-NUTshadow}
We will now discuss the properties of the sets  $\mathbb{T}_\pm(p)$ in Kerr-Newman-Taub-NUT. Note that the equations of motion for $r$ \eqref{eq:radial} and $\theta$ \eqref{eq:theta} have two solutions that differ only by a sign for a fixed combination of $E, L_z, K$. Therefore we know that the trapped sets will have a reflection symmetry across the $k_1=0$ and the $k_2=0$ planes. A sign change in $k_2$ maps both sets $\mathbb{T}_\pm(p)$ to themselves, while a sign flip in $k_1$ maps $\mathbb{T}_+(p)$ to $\mathbb{T}_-(p)$ and vice versa. We start by analyzing the sets for points of symmetry
\begin{definition}A Point of symmetry is a point for which there exists a one parameter family of diffeomorphism with closed orbits, which all leave the point itself invariant.\end{definition}
\begin{lemma}\label{lem:symmetry}The sets $\mathbb{T}_+(p)$ and $\mathbb{T}_-(p)$ are circles on the celestial sphere of any timelike observer at any regular point of symmetry in the exterior region of any subextremal Kerr-Newman-Taub-NUT spacetime. 
\end{lemma}
\begin{proof}To determine the structure of $\mathbb{T}_\pm(p)$ we observe that when we pick a future/past trapped direction and apply the diffeomorphism, the spacial directions of $TM|_p$ are rotated around the vector pointing along the axis left invariant by the diffeomorphism. Therefore the future/past trapped direction traces proper circles on the celestial sphere. Therefore the future and past trapped set at such a point $p$ always correspond to a collection of circles independent of the details of the manifold or the location of $p$ therein. We are now going to show that in the spacetimes under consideration here $\mathbb{T}_\pm(p)$ consist of exactly one circle. For the spherically symmetric spacetimes this is a well known fact. For an observer located at a regular point on the rotation axis of Kerr-Newmann-Taub-NUT black holes we can apply the following argument. Note that for $l\neq0$ and $a\neq0$ we have to choose $C=\pm1$ for the procedure to apply to the regular part of the rotation axis in these cases. For all other values of $C$ both parts of the rotation axis are singular. Hence the discussion here does not apply to those cases.\\
From equation \eqref{eq:theta} it is clear that null geodesics that can reach the rotation axis have to have $L_z=0$. For the case $L_z=0$ it is clear that there exists only one value of $K_E^{trapp}(L_z=0)$ and $r_{trapp}(L_z=0)$. To treat an observer on the rotation axis we need to introduce a new coordinate system which covers the axis. We will use Cartesian coordinates $(t, x, y, z)$, which are related to the Boyer-Lindquist coordinates $(t, r, \theta, \phi)$ by the following relations:
\begin{equation}
\begin{aligned}
t &= t \\
x &= r \sin(\theta) \cos(\phi) \\
y &= r \sin(\theta) \sin(\phi) \\
z &= r \cos(\theta) 
\end{aligned}
\end{equation}
Then the following set is an orthonormal thetrad on the rotation axis ($x=y=0$):
\begin{align}\label{eq:tetradct}
\tilde e_0 &= \left.-\sqrt{\frac{z^2+(a+l)^2}{z^2-2m z +a^2 + Q^2-l^2 }} \partial_t \right|_p ,&\qquad 
\tilde e_1&=\left. \frac{z}{\sqrt{z^2+(a+l)^2}}\partial_x \right|_p, \\ \nonumber
\tilde e_2&=\left. \frac{z}{\sqrt{z^2+(a+l)^2}}\partial_y \right|_p,&\qquad 
\tilde e_3&=\left. \sqrt{\frac{z^2-2m z +a^2 + Q^2-l^2}{z^2+(a+l)^2}}\partial_z \right|_p. 
\end{align}
As $\tilde e_3$ points along the rotation axis and is thus left invariant under a rotation of the manifold, we know that along the trapped set $k_x^2+k_y^2=const.$ will be satisfied. Calculating Carter's constant from the tangent vector on the rotation axis we see that it is given by:
\begin{equation}
K = \sigma_{\mu \nu}   \dot \gamma^\mu \dot \gamma^\nu = (1-k_z^2) ((a+l)^2 + z^2)
\end{equation}
In the above expression, $\sigma_{\mu \nu}$ is the Killing tensor expressed in Cartesian coordinates, on the rotation axis ($x=y=0$).We see immediately that on the celestial sphere there exists at most two values of $k_z$ such that:
\begin{equation}
\left. K_E(k_z) \right|_p=\left. \frac{K(k_z)}{E^2(\tilde e_0)}\right|_p=K_E^{trapp}(L_z=0)
\end{equation}
These correspond to the future and the past trapped set. If the two solutions coincide we are at $z=r_{trapp}(L_z=0)$ and the directions are both future and past trapped. It remains to show that there will always be at least one value of $k_z$, such that the condition for future/past trapping is satisfied. We know that $k_z=-1$ always hits the horizon, while $k_z=1$ always escapes to infinity and the infalling and outgoing sets are open due to the continuous dependence on initial data for solutions to the geodesic equation. Therefore, the trapped sets have to be non-empty.
\end{proof}
In the following we will use the parametrization of \cite{grenzebach_photon_2014}. We introduce the coordinates $\rho \in[0,\pi]$ and $\psi \in[0,2\pi)$ on the celestial sphere. Thus \eqref{eq:tangentplus} can be written as:
\begin{equation}\label{eq:tangentsph}
\dot\gamma( \rho,\psi)|_p =\alpha ( e_0+ \cos(\rho) e_1+ \sin(\rho) \cos(\psi) e_2+ \sin(\rho)\sin(\psi) e_3)
\end{equation}
The principal null direction towards the black hole is given by $\rho=\pi$. Following \cite{grenzebach_photon_2014}  one finds the following parametrization of the celestial sphere in terms of constants of motion:
\begin{subequations}\label{eq:comonsphere}
\begin{align}
\sin(\psi)&= \left.\frac{\tilde L_E +a \cos^2 (\theta)+2l\cos(\theta)}{\sqrt{K_E}\sin(\theta)}\right|_{\theta(p)}\\
\sin(\rho)&= \left. \frac{\sqrt{\Delta K_E}}{r^2 +l^2 - a\tilde L_E}\right|_{r(p)}
\end{align}
\end{subequations}
where
\begin{equation}
\tilde L_E=L_E -a +2lC
\end{equation} 
By plugging the relations \eqref{eq:altcomtrap} in the above equations we obtain a parametrization of the future and past trapped set in terms of the radius to which a particular future trapped direction is converging to. We use $x$ to parametrize the trapped sets with  $x=r_{trapp}$.  We will show in the following proof that the parameter $x$ is restricted to the interval $[r_{min}(\theta),r_{max}(\theta)]$. Here $r_{min}(\theta)$ and $r_{max}(\theta)$ are given as the intersection of a cone of constant $\theta$ with the boundary of the area of trapping. The parametrization is given by: 
\begin{subequations}\label{eq:parametrized}
\begin{align}
f(x)&:=\sin(\psi)= \frac{\Delta'(x)\{x^2 +(l+a \cos [\theta(p)])^2\}-4 x \Delta(x)}{4ax\sqrt{\Delta(x)} \sin(\theta(p))}\label{eq:psitox}\\ \nonumber\\
h(x)&:=\sin(\rho)= \frac{4x\sqrt{\Delta(r(p))\Delta(x)}}{\Delta'(x)(r(p)^2-x^2) + 4x \Delta(x)}\label{eq:rhotox} \\ \nonumber 
\end{align}
\end{subequations}
We are now ready to prove our main Theorem.
\begin{thm}\label{thm:1}
The sets $\mathbb{T}_+(p)$ and $\mathbb{T}_-(p)$ are smooth curves on the celestial sphere of any timelike observer at any point in the exterior region of any subextremal Kerr-Newman-Taub-NUT spacetime. 
\end{thm}
\begin{proof}
 We start by analyzing the right-hand side of \eqref{eq:psitox}:
\begin{equation}
\frac{d f(x)}{d x}=\frac{\{x^2 + (l+a \cos [\theta])^2\}((M-x)^3-M(M^2-a^2-Q^2+l^2))}{2 a x^2 \Delta^{3/2}\sin(\theta)}
\end{equation} 
which is strictly negative for $x\in(r_+,\infty)$. Further the limit of the right hand side of \eqref{eq:psitox} is $\infty$ for $x\rightarrow r_+$ and $-\infty$ for $x\rightarrow\infty$. Therefore, the function $f$ is strictly monotone in the interval  $x\in(r_+,\infty)$ and, hence invertible. Then, $x(\psi)=f^{-1}(\sin(\psi))$ is a smooth function of $\psi$ with extrema at the extremal points of $\sin(\psi)$. As was observed in \cite{grenzebach_photon_2014} the minimum $x=r_{min}(\theta(p))$ at $\psi=\pi/2$ and the maximum of $x=r_{max}(\theta(p))$ at $\psi=3\pi/2$ correspond exactly to the intersections of a cone with constant $\theta$ with the boundary of the region of trapping. This can be seen by setting the left hand side of \eqref{eq:psitox} equal to $\pm1$, and comparing to \eqref{eq:areaoftrapp}. So we have that  $x(\psi)\in[r_{min}(\theta(p)),r_{max}(\theta(p))]$ for all values of $\theta(p)$.

Now we take a look at the right hand side of equation \eqref{eq:rhotox}:
\begin{equation}
\label{eq:hxder}
\frac{d h(x)}{d x}= \frac{2(r^2-x^2)\Delta(r)((x-M)^3+ M(M^2-a^2-Q^2 +l^2))}{\sqrt{\Delta(x)\Delta(r)}((r^2-x^2)\frac{\Delta'(x)}{2}+2x\Delta(x))^2} .
\end{equation} 
This is positive when $x<r(p)$ and negative when $x>r(p)$. The denominator never vanishes for $x\in(r_+,\infty)$ because:
\begin{equation}
(4x\Delta(x)+ (r(p)^2-x^2)\Delta'(x))|_{\{r(p)=r_+, x=r_+\}}=0
\end{equation}
and 
\begin{align}
\frac{d}{dx}(4x\Delta(x)+ (r(p)^2-x^2)\Delta'(x))&=2(3x^2-6Mx+2(a^2-l^2+Q^2)+r(p)^2)>0, \label{eq:first}\\
\frac{d}{dr(p)}(4x\Delta(x)+ (r(p)^2-x^2)\Delta'(x))&=2r(p)\Delta'(x)>0,
\end{align}
where we used $r(p)>r_+>M>\sqrt{a^2-l^2+Q^2}$ in \eqref{eq:first}. \\
If we set $x=r(p)$ in \eqref{eq:rhotox} then the right-hand side is equal to $1$. Furthermore in any of the limits $r(p)\rightarrow r_+$, $r(p)\rightarrow \infty$, $x\rightarrow r_+$, and as $x\rightarrow\infty$ it goes to zero. 
\begin{case}If $p\notin \mathcal{A}$ hence if $r(p)\notin[x_{min}(\theta(p)),x_{max}(\theta(p))]$ then the two functions
\begin{align}
\rho_1(\psi)=\arcsin(h(x(\psi)))&: [0,2\pi)\rightarrow [\rho_{1_{min}},\rho_{1_{max}}]\subset \left(0,\frac{\pi}{2}\right)\\
\rho_2(\psi)=\pi-\arcsin(h(x(\psi)))&: [0,2\pi)\rightarrow[\rho_{2_{min}},\rho_{2_{max}}]\subset\left(\frac{\pi}{2},\pi\right)
\end{align}
are both smooth with $\rho_1(0)=\rho_1(2\pi)$ and $\rho_2(0)=\rho_2(2\pi)$. 
If $p$ is between the region of trapping and the asymptotically flat end, the function $\rho_2(\psi)$ corresponds to $\mathbb{T}_+(p)$ and $\rho_1(\psi)$ corresponds to $\mathbb{T}_-(p)$. Because $(\pi/2,\pi]$ corresponds to the geodesic with $\dot r< 0$. If $p$ is between the region of trapping and the horizon then the role of $\rho_1(\psi)$ and $\rho_2(\psi)$ are switched.
\end{case}
\begin{case}If $p\in \mathcal{A}$ we need to do some extra work. For simplicity we only consider the interval $\psi\in[\pi/2,3\pi/2]$ as the rest follows by symmetry of $\sin(\psi)$ in $[0,\pi]$ across $\pi/2$ and in $[\pi,2\pi]$ across $3\pi/2$. We define:
\begin{equation}
    \psi_0(r(p))= \pi-\arcsin\left(\frac{\Delta'(r(p))\{r(p)^2 +(l+a \cos [\theta(p)])^2\}-4 r(p) \Delta(r(p))}{4r(p)\sqrt{\Delta(r(p))}a\sin(\theta(p))}\right).
\end{equation}
The two functions
\begin{align}
\rho_3(\psi)&=\begin{cases}
\arcsin(h(x(\psi))) &\text{ if } \psi\in[\pi/2,\psi_0(r(p))]\\
\pi-\arcsin(h(x(\psi)))&\text{ if } \psi\in (\psi_0(r(p)),3\pi/2]
\end{cases}\\
\rho_4(\psi)&=\begin{cases}
\pi-\arcsin(h(x(\psi))) &\text{ if } \psi\in[\pi/2,\psi_0(r(p))]\\
\arcsin(h(x(\psi)))&\text{ if } \psi\in (\psi_0(r(p)),3\pi/2]
\end{cases}
\end{align}
are then smooth on $[\pi/2,3\pi/2]$. For a proof see Appendix \ref{app:A} and note that at $\psi_0$, $h(x(\psi))$ satisfies the conditions required in the appendix. Since $p\in \mathcal{A}$ we have that $x_{min}(\theta(p))<r(p)<x_{max}(\theta(p))$. Therefore the geodesic on the celestial sphere parametrized by $x_{max}(\theta(p))$ has to have $\dot r>0$ and thus has to be in $[0,\pi/2)$.  On the other hand the geodesic on the celestial sphere parametrized by $x_{min}(\theta(p))$ has to have $\dot r<0$ and thus has to be in $(\pi/2,\pi]$. In fact by the monotonicity of the right hand side of \eqref{eq:psitox} and the fact that $x(\psi_0)=r(p)$ we know that for $\psi \in [\pi/2,\psi_0)$ we have $x(\psi)<r(p)$ and for $\psi \in (\psi_0,3\pi/2]$ we have $x(\psi)>r(p)$. Thus we can conclude that for $p\in \mathcal{A}$, $\rho_4$ corresponds to $\mathbb{T}_+(p) $ and $\rho_3$ corresponds to $ \mathbb{T}_-(p)$ and thus both sets are smooth.\end{case}
\begin{case}In the special case when $r(p)= x_{max}(\theta(p))$ or $r(p)= x_{min}(\theta(p))$ the functions $\rho_1$ and $\rho_2$,  which describe $\mathbb{T}_\pm(p)$, do reach $\rho=\pi/2$ at $\psi=3\pi/2$ (for $r(p)= x_{max}(\theta(p))$) or $\psi=\pi/2$ (for $r(p)= x_{min}(\theta(p))$) respectively. However since in these cases we have that
\begin{equation}
    \frac{d^2}{d\psi^2}(h(x(\psi)))=0 
\end{equation}
the two sets meet at this point tangentially and do not cross over into the other hemisphere.\end{case}
Together with Lemma \ref{lem:symmetry} this concludes the proof. 
\end{proof}
\begin{remark}
In \cite{grenzebach_photon_2014} it was observed that $\rho_{max}$ of $\mathbb{T}_+(p) $ always corresponds to the trapped geodesic with $x_{min}(\theta(p))$ and $\rho_{min}$ of $\mathbb{T}_+(p) $ always corresponds to the trapped geodesic with $x_{max}(\theta(p))$ . When $p$ is outside the region of trapping $h(x)|_{x_{max}}$ is a local maximum of $h(x(\psi))$ (as a function of $\psi$) and  $h(x)|_{x_{min}}$ is a local minimum of $h(x(\psi))$. When $p$ is between the region of trapping and the horizon $h(x)|_{x_{max}}$ is a local minimum of $h(x(\psi))$ and  $h(x)|_{x_{min}}$ is a local maximum of $h(x(\psi))$. Since outside $\mathbb{T}_+(p) $ is always described by $\rho_2(\psi)$ and inside by $\rho_1(\psi)$, $\rho_{min}$ then always corresponds to $x_{min}$ and $\rho_{max}$ always corresponds to $x_{max}$. This also holds for $p\in\mathcal{A}$. For $\mathbb{T}_-(p) $ the correspondence is switched.
\end{remark}
\begin{remark}
We have only proved Theorem \ref{thm:1} for one standard observer at any particular point. However since any other observer at this point is related to the standard observer by a Lorentz transformation and the Lorentz transformations act as conformal transformations on the celestial sphere \cite[p.14]{penrose_spinors_1987}, the Theorem indeed holds for any observer. In \cite{grenzebach_aberrational_2015} the quantitative effect on the shape of the shadow of boosts in different directions are discussed.
\end{remark}
\begin{remark}
The parametrization for the trapped set on the celestial sphere of any standard observer in \cite{grenzebach_photon_2014,grenzebach_photon_2015} was derived for a much more general class of spacetimes. Therefore  Theorem \ref{thm:1} might actually hold for these cases as well. However this is beyond the scope of this paper.
\end{remark}
From Theorem \ref{thm:1} we immediately get the following Corollary:
\begin{cor}
For any observer at any regular point $p$ in the exterior region of a subextremal Kerr-Newman-Taub-NUT spacetime away from the axis of symmetry we have that for any $k\in \mathbb{T}_+(p)$ and any $\epsilon>0$
\begin{itemize}[noitemsep]
\item $B_\epsilon(k)\cap \Omega_{\mathcal{H}^+}(p)\neq \emptyset$
\item $B_\epsilon(k)\cap \Omega_{\scri^+}(p)\neq \emptyset$.
\end{itemize}
\end{cor}
\noindent So if we interpret the celestial sphere as initial data space for null geodesics starting at $p$, the Corollary is a coordinate independent formulation of the fact that trapping in the exterior region of subextremal Kerr-Newman-Taub-NUT black holes is unstable.\\

\section{Conclusion}
Despite the fact that trapping in a Kerr-Newman-Taub-NUT spacetime is much more complicated than in Schwarzschild, we showed in the present work that the topological structure of the future and past trapped set at any point in the exterior region in Kerr-Newman-Taub-NUT is in fact simple and identical to the situation in Schwarzschild. Even though the qualitative features of $\mathbb{T}_\pm(p)$ do not change under a change of parameters, the quantitative features do.\\
In \cite{hioki_measurement_2009,li_measuring_2014} it is discussed what information can be read of from the shadow at infinity. For points inside the manifold as considered in this work, this question will be answered in our upcoming paper \cite{shadows}. The here presented result will be the foundation for our upcoming work \cite{shadows}.

\subsection*{Acknowledgements} We are grateful to Marc Mars, J\'er\'emie Joudioux,  Lars Andersson, Volker Perlick and Siyuan Ma for helpful discussions and their comments on the manuscript. Special thanks goes to Blazej Ruba for his help with an earlier draft of this paper. 
\appendix
\section{}\label{app:A}
Let $f(x)$ be a smooth function on $[-1,1]$ vanishing at the boundary points with a unique maximum with value $1$ at zero. Furthermore, we consider that $f(0)=1$, $f'(0)=0$ and $f''(0)<0$. Then, we define:
\begin{align}
g_{1a}(x)&=\arcsin(f(x)): & [-1,0)\rightarrow&[0,\pi/2)\\
g_{2a}(x)&=\pi-\arcsin(f(x)):& [-1,0)\rightarrow&(\pi/2,\pi]\\
g_{1b}(x)&=\arcsin(f(x)):& (0,1]\rightarrow&[0,\pi/2)\\
g_{2b}(x)&=\pi-\arcsin(f(x)):& (0,1]\rightarrow&(\pi/2,\pi]
\end{align}
Note that $g_{1a/b}'(x)=-g_{2a/b}'(x)$. By standart analytic arguments one gets that $d/dx (\arcsin(f(x)))|_{x=0}= \sqrt{-f''(0)}$. Note that on $[-1,0)$ the derivative of $g_1(x)$ is positive while on $(0,1]$ it is negative. Together, this gives us that the function
\begin{equation}
    g(x)=\begin{cases}
    g_{1a}(x)& \text{ if } x\in [-1,0)\\
    \pi/2&  \text{ if } x=0\\
    g_{2b}(x)&  \text{ if } x\in (0,1]
    \end{cases}
\end{equation}
is smooth at $x=0$ and therefore on $[-1,1]$, with $d/dx(g(x))|_{x=0}=\sqrt{-f''(0)}$. 



\newcommand{\arxivref}[1]{\href{http://www.arxiv.org/abs/#1}{{arXiv.org:#1}}}
\newcommand{\mnras}{Monthly Notices of the Royal Astronomical Society}
\newcommand{\prd}{Phys. Rev. D}
\newcommand{\apj}{Astrophysical J.}

\bibliographystyle{amsplain}
\bibliography{shadow,kerr}

\end{document}